\makeatletter
\AtBeginDocument{%
}
\makeatother

\documentclass{ifacconf}

\usepackage{natbib}
\usepackage{amsmath,amssymb,amsfonts}
\let\theoremstyle\relax
\usepackage{amsthm}
\usepackage{array}
\usepackage{bm}
\usepackage{booktabs}
\usepackage{float}
\usepackage{graphicx}
\usepackage{multirow}
\usepackage{pgf,tikz}
\usepackage{siunitx}
\usepackage{textcomp}
\usepackage[table]{xcolor}
\usepackage{wrapfig}

\usetikzlibrary{arrows}
\usetikzlibrary{shapes}

\newtheorem{theorem}{Theorem}

\newtheorem{remark}{Remark}

\DeclareMathAlphabet{\pazocal}{OMS}{zplm}{m}{n}

\begin{document}
\begin{frontmatter}

\title{Online Bayesian Learning of Agent Behavior in Differential Games}

\author[]{Francesco Bianchin, Robert Lefringhausen, Sandra Hirche} 

\address{School of Computation, Information and Technology \\ (e-mail:
{\tt\small\{francesco.bianchin, robert.lefringhausen, sandra.hirche\}@tum.de}).}

\begin{abstract}                
This work introduces an online Bayesian game-theoretic method for behavior identification in multi-agent dynamical systems. By casting Hamilton–Jacobi–Bellman optimality conditions as linear-in-parameter residuals, the method enables fast sequential Bayesian updates, uncertainty-aware inference, and robust prediction from limited, noisy data—without history stacks. The approach accommodates nonlinear dynamics and nonquadratic value functions through basis expansions, providing flexible models. Experiments, including linear–quadratic and nonlinear shared-control scenarios, demonstrate accurate prediction with quantified uncertainty, highlighting the method’s relevance for adaptive interaction and real-time decision making.
\end{abstract}


\end{frontmatter}

\section{Introduction}
In multi-agent settings, effective interaction increasingly depends on the ability to infer, predict, and adapt to the objectives of other agents. Collaborative, competitive, and mixed-autonomy scenarios—ranging from shared-control settings to multi-agent coordination and broader human–robot interaction—can all be modeled as differential games, in which each agent optimizes its own objective while interacting within a shared dynamical environment. Game-theoretic models provide a principled framework for describing such interactions \citep{isaacs1965differential,basar_book}, and have proven effective for generating adaptive and anticipatory strategies across a variety of domains \citep{MUSIC202010216,Li_HRI,JAS-2019-0028}.

A central challenge is that agent objectives are often unknown or only partially specified. Controllers that assume perfect knowledge of these objectives can perform poorly when agents deviate from the expected behavior, leading to degraded performance or unsafe outcomes. As a result, many practical approaches avoid constructing explicit models of human intent, opting instead for adaptive or heuristic strategies that respond to observed behavior \citep{Li_HRI,MUSIC202010216}. While successful in specific settings, these methods provide limited interpretability and lack robustness under irregular or hard-to-model behaviors.

Inverse Optimal Control \citep{IRL_observer} and Inverse Reinforcement Learning \citep{IRL_review} aim to recover underlying objective functions from observations, offering structured models of decision making. Extensions to multi-agent settings via inverse differential games enable identification of individual objectives in coupled systems \citep{ILQG_SharedControl,molloy2022inverse}. However, key challenges remain. Updates must be fast enough for online use \citep{IRL_observer, ILQG_SharedControl}; moreover, real-world applications frequently involve nonlinear dynamics and nonquadratic value functions, extending beyond classical linear–quadratic (LQ) formulations \citep{engwerda, FB_inverse_LQDG, bianchin2025settheoreticrobustcontrolapproach}. The general setting has been handled in different works by expressing the Hamilton–Jacobi–Bellman (HJB) optimality conditions underlying demonstrations as linear-in-parameter residuals, but existing methods either require extensive information on the agents' objectives \citep{VAMVOUDAKIS20111556}—incompatible with settings involving unknown agents—or focus solely on single-agent problems \citep{Self2018OnlineIR}. Most importantly, inference must be robust to disturbances and behavioral variability. For example, human behavior routinely departs from nominal patterns, yet deterministic predictors cannot accommodate such deviations. Ensuring safety and reliability therefore requires uncertainty-aware inference that explicitly represents variability rather than assuming a single predicted trajectory. In addition, predictions should remain informative with limited initial data.

The main contribution of this paper is a Bayesian framework for inferring agent objectives in interactive settings from observed behavior. The method models approximate satisfaction of HJB optimality conditions and enables fast sequential updates through conjugate Gaussian inference, producing uncertainty-aware estimates of agent objectives at control frequencies. This formulation provides principled regularization through the formulation of priors, avoids the history-stack requirements common in prior approaches \citep{IRL_observer}, and yields informative predictions from limited data. The quantified uncertainty directly supports conservative or exploratory policies, enabling confidence-aware decision making in multi-agent interaction. The framework also extends naturally beyond LQ assumptions: nonlinear dynamics and general value functions are accommodated via basis expansions, offering flexible yet interpretable models capable of capturing parametric and behavioral variation over time. We validate the approach in both LQ and nonlinear shared-control scenarios, demonstrating its practical relevance and applicability.

The remainder of the paper is structured as follows. Section~II formulates the differential game model and the associated inverse problem. Section~III introduces the approximate HJB regression framework. Section~IV presents the online Bayesian estimator and its role in probabilistic prediction. Section~V illustrates the approach in both linear–quadratic and nonlinear shared-control scenarios. Section~VI concludes the paper.

\section{Problem Formulation}

We consider a two-player, infinite-horizon, nonzero-sum differential game with dynamics
\begin{equation}
    \dot{\bm{x}}
    = \bm{f}(\bm{x})
    + \bm{g}_1(\bm{x})\,\bm{u}_1
    + \bm{g}_2(\bm{x})\,\bm{u}_2 ,
    \label{eq:dynamics}
\end{equation}
where $\bm{x}\in\mathbb{R}^{n_x}$ and $\bm{u}_i\in\mathbb{R}^{n_{u_i}}$ denote the control inputs of players $i\in\{1,2\}$. 
Each player minimizes the infinite-horizon cost
\begin{equation}
    J_i(\bm{x}_0,\bm{u}_1,\bm{u}_2)
    = \int_{0}^{\infty}
        \Big(
            Q_i(\bm{x}(t))
            + \bm{u}_i(t)^\top \bm{R}_i \bm{u}_i(t)
        \Big) dt ,
    \label{eq:cost}
\end{equation}
where the state cost $Q_i:\mathbb{R}^{n_x}\!\to\!\mathbb{R}_{\ge 0}$ and the positive-definite matrix $\bm{R}_i$ are unknown to the learner (which may be an external estimator or one of the players).

\begin{remark}
The two-player setup is used for notational clarity; all definitions extend directly to $N$-player differential games with dynamics $\dot{\bm{x}}=\bm{f}(\bm{x})+\sum_{i=1}^N \bm{g}_i(\bm{x})\bm{u}_i$ and analogous cost functionals.
\end{remark}
We assume that the agents interact through stabilizing stationary feedback laws $\bm{u}_i=\bm{\mu}_i(\bm{x})$, yielding value functions
\begin{equation}
    V_i(\bm{x}(t))
    = \int_{t}^{\infty}
        \Big(
            Q_i(\bm{x}(\tau))
            +
            \bm{\mu}_i(\bm{x}(\tau))^\top
            \bm{R}_i\,
            \bm{\mu}_i(\bm{x}(\tau))
        \Big) d \tau .
    \label{eq:value}
\end{equation}

We assume that the agents operate at a feedback Nash equilibrium \citep{basar_book}, namely a condition in which no component of the strategy pair $(\boldsymbol{\mu}_1^*, \boldsymbol{\mu}_2^*)$ can be changed unilaterally without increasing the corresponding player’s cost.
At any such equilibrium, $(V_i,Q_i,\bm{R}_i)$ satisfy the coupled Hamilton--Jacobi--Bellman equations
\begin{align}
0 &=
Q_i(\bm{x})
+ \bm{u}_i^\top \bm{R}_i \bm{u}_i
+ \nabla V_i(\bm{x})^\top
  \Big(
      \bm{f}(\bm{x})
      + \bm{g}_1(\bm{x})\bm{u}_1
  \nonumber\\[-2mm]
&\qquad\qquad\qquad\qquad
      +\,\bm{g}_2(\bm{x})\bm{u}_2
  \Big)
=: \pazocal{H}_i(\bm{x},\bm{u}_1,\bm{u}_2),
\label{eq:hjb_general}
\end{align}
and the stationarity of $\pazocal{H}_i$ with respect to $\bm{u}_i$ yields the equilibrium feedback law
\begin{align}
    \bm{\mu}_i^*(\bm{x})
    = -\tfrac{1}{2}\,\bm{R}_i^{-1}\bm{g}_i(\bm{x})^\top \nabla V_i(\bm{x}),
    \label{eq:feedback}
\end{align}
which, by $\bm{R}_i \succ 0$, uniquely characterizes the optimal action for player~$i$.

In the inverse setting considered here, only discrete-time samples of the 
agents’ closed-loop trajectories 
$(\bm{x}(t_k), \dot{\bm{x}}(t_k), \allowbreak \bm{u}_1(t_k), \bm{u}_2(t_k))$ are available, 
while the objective components $(Q_i, V_i, \bm{R}_i)$ are unknown.  
We assume measurements are collected at times $t_k = k\Delta t$, and 
$\dot{\bm{x}}(t_k)$ is obtained either from the known 
dynamics~\eqref{eq:dynamics} or from filtered finite differences.  
The observed behavior is assumed to approximately satisfy the HJB 
relations~\eqref{eq:hjb_general} and the feedback structure~\eqref{eq:feedback}.

\textbf{Problem statement.}
Given dynamics~\eqref{eq:dynamics} and observed trajectories, infer probabilistic models of the unknown objective components $(Q_i, V_i, \bm{R}_i)$ whose posterior realizations approximately satisfy the Hamilton--Jacobi--Bellman relations~\eqref{eq:hjb_general} together with the corresponding feedback conditions~\eqref{eq:feedback}. The aim is to obtain predictive distributions over players' admissible equilibrium behavior across interactive settings.

To obtain a tractable inference procedure, we next introduce a feature-based parametrization of the objectives that makes the HJB relations linear in the decision parameters.

\section{Feature-Based Inverse Differential Games}
We now introduce the feature-based parametrization used to represent the 
unknown objectives $(V_i, Q_i, \bm{R}_i)$.  
Let $\Omega \subset \mathbb{R}^{n_x}$ be a compact subset of the state 
space containing all observed trajectories. On~$\Omega$, stabilizing 
feedback laws yield value functions $V_i$ that are locally Lipschitz and 
differentiable almost everywhere. Under standard regularity conditions, 
$V_i$ belongs to a suitable Sobolev space on~$\Omega$, and classical 
approximation results 
(e.g., Weierstrass and universal approximation theorems 
\citep{ABUKHALAF2005779, adams2003sobolev}) ensure that both $V_i$ and 
its gradient can be approximated arbitrarily well on~$\Omega$.  
Accurate approximation of $\nabla V_i$ is particularly important, as it 
enters the HJB relations directly and hence determines the Hamiltonian 
residuals used for inverse inference.

To obtain a tractable, linear-in-parameters representation of the unknown 
objectives, we approximate the value and state cost functions using 
differentiable feature maps
\[
    \phi_{V_i}:\Omega\to\mathbb{R}^{p_i}, \qquad
    \phi_{Q_i}:\Omega\to\mathbb{R}^{s_i},
\]
with typical choices including polynomial features, radial basis functions, 
or other standard universal approximators.  
The feature sets may incorporate task- or domain-specific structure (e.g., known 
symmetries or dominant state interactions) and are selected to provide smooth 
approximations compatible with the HJB relations.
The control cost does not require approximation: since 
$\bm{u}_i^\top \bm{R}_i \bm{u}_i$ is quadratic in $\bm{u}_i$, it admits an 
exact linear parameterization through
\[
    \phi_{R_i}(\bm{u}_i)
    :=
    \operatorname{vec}\!\left(\bm{u}_i \bm{u}_i^\top\right)
    \in\mathbb{R}^{z_i}.
\]

Under these representations, the objectives take the form
\begin{subequations}\label{feat_approx}
\begin{align}
    V_i(\bm{x})
        &= \bm{W}_{V_i}^\top \phi_{V_i}(\bm{x})
           + \varepsilon_{V_i}(\bm{x}), \\[1mm]
    Q_i(\bm{x})
        &= \bm{W}_{Q_i}^\top \phi_{Q_i}(\bm{x})
           + \varepsilon_{Q_i}(\bm{x}), \\[1mm]
    \bm{u}_i^\top \bm{R}_i \bm{u}_i
        &= \bm{W}_{R_i}^\top \phi_{R_i}(\bm{u}_i),
\end{align}
\end{subequations}
where terms $\varepsilon_{V_i}$ and $\varepsilon_{Q_i}$ capture the 
approximation error, which decreases for standard universal or 
increasing-order feature families as $(p_i,s_i)$ grow.  
In contrast, the control cost is represented exactly, as 
$\bm{u}_i^\top \bm{R}_i \bm{u}_i$ is linear in the entries of $\bm{R}_i$.
The parameters 
$\bm{W}_{V_i}\!\in\!\mathbb{R}^{p_i}$, 
$\bm{W}_{Q_i}\!\in\!\mathbb{R}^{s_i}$, and 
$\bm{W}_{R_i}\!\in\!\mathbb{R}^{z_i}$ 
encode the unknown objectives; recovering them is therefore equivalent to 
solving the inverse differential game.  
Observed trajectory samples yield data 
through the HJB and feedback relations, and the aim is to infer posterior 
distributions over the weights such that
\[
    \pazocal{H}_i(\bm{x},\bm{u}_1,\bm{u}_2;
        \bm{W}_{V_i},\bm{W}_{Q_i},\bm{W}_{R_i})
    \approx 0
\]
along the observed trajectories.

\subsection{HJB conditions in feature space.}
Substituting~\eqref{feat_approx} into the stationary HJB 
equation~\eqref{eq:hjb_general} and using the dynamics~\eqref{eq:dynamics} 
yields the linear-in-parameters condition
\begin{align}
    0 &=
    \bm{W}_{Q_i}^\top \phi_{Q_i}(\bm{x})
    + \bm{W}_{R_i}^\top \phi_{R_i}(\bm{u}_i)
    + \bm{W}_{V_i}^\top
        \nabla\phi_{V_i}(\bm{x})\,\dot{\bm{x}}
    + \eta_{i}^{\mathrm{HJB}},
    \label{eq:hjb_linear_param}
\end{align}
where $\eta_{i}^{\mathrm{HJB}}$ collects the approximation errors 
$\varepsilon_{V_i}$, $\varepsilon_{Q_i}$ and their gradients.  
These terms are treated as residual noise in the regression model below.

Similarly, substituting into the feedback law~\eqref{eq:feedback} gives
\begin{align}
    \bm{u}_i
    = -\tfrac{1}{2}\bm{R}_i^{-1}
      \bm{g}_i(\bm{x})^\top
      \nabla\phi_{V_i}(\bm{x})^\top \bm{W}_{V_i}
      + \eta_i^{\mathrm{fb}},
    \label{eq:feedback_linear_param}
\end{align}
where $\eta_i^{\mathrm{fb}}$ captures the corresponding approximation error.

Collecting all unknown coefficients in
\[
    \bm{W}_i =
    \big[
        \bm{W}_{V_i}^\top \;\;
        \bm{W}_{Q_i}^\top \;\;
        \bm{W}_{R_i}^\top
    \big]^\top,
\]
the relations~\eqref{eq:hjb_linear_param}–\eqref{eq:feedback_linear_param} 
yield linear constraints on $\bm{W}_i$ up to residual noise terms.  
Since inverse differential-game objectives are identifiable only up to a
positive scale factor~\citep{IRL_observer}, we fix $R_{i,[11]}$ and remove this
entry, yielding the reduced vector 
\[
    \bm{W}_i^-=
    \big[
        \bm{W}_{V_i}^\top\;\;
        \bm{W}_{Q_i}^\top\;\;
        (\bm{W}_{R_i}^-)^\top
    \big]^\top.
\]
Constant offsets in $V_i$ do not affect the feedback law and are omitted, so
$\bm{W}_i^-$ represents one element of the equivalence class consistent with
the observed behavior.

This leads to the regression model  
\begin{align}
    \bm{y}_i = \bm{\Phi}_i \bm{W}_i^- + \bm{\eta}_i ,
    \label{eq:regression_model}
\end{align}
where $\bm{\eta}_i$ collects residual terms arising from finite-dimensional 
approximation~\eqref{feat_approx}, numerical differentiation, and deviations from exact HJB 
optimality along the observed trajectories.  
The regressor and target follow directly from the HJB stationarity and feedback optimality conditions, with the regressor matrix given by
\begin{align}
\bm{\Phi}_i =
\begin{bmatrix}
    (\nabla\phi_{V_i}\dot{\bm{x}})^\top
        &
        \phi_{Q_i}^\top
        &
        \phi_{R_i}^{-,\top}
      \\[4pt]
    \bm{g}_i^\top \nabla\phi_{V_i}^\top
        &
        \mathbf{0}_{m_i\times s_i}
        &
        \begin{bmatrix}
            \mathbf{0}_{1\times(z_i-1)} \\
            2\,\mathrm{diag}(u_{i,2:n_{u_i}})
        \end{bmatrix}
\end{bmatrix},
\end{align}
where the bottom-right block reflects the gradient of 
$\bm{u}_i^\top \bm{R}_i \bm{u}_i$ with respect to the remaining entries 
of $\bm{R}_i$ after removing the fixed element $R_{i,[11]}$.  
The regression target becomes
\begin{align}
    \bm{y}_i =
    \begin{bmatrix}
        - R_{i,[11]} u_{i,1}^2 \\
        - 2 R_{i,[11]} u_{i,1} \\
        \mathbf{0}_{(n_{u_i}-1) \times 1}
    \end{bmatrix},
\end{align}
recalling that $R_{i,[11]}$ is fixed to resolve the scale ambiguity in the inverse HJB formulation.

The first row enforces the HJB stationarity condition~\eqref{eq:hjb_linear_param},  
while the remaining rows encode the optimality conditions~\eqref{eq:feedback_linear_param}.  
The linear form~\eqref{eq:regression_model} enables conjugate Bayesian 
updates for real-time, uncertainty-aware inference of the players' objectives.

\section{Bayesian Online Estimation of Player Objectives}
The linear relation~\eqref{eq:regression_model} enables probabilistic 
inference over the reduced parameter vector $\bm{W}_i^-$.  
A Bayesian treatment is natural here: it allows prior structural information 
about the objectives to be incorporated and maintains an explicit 
representation of uncertainty, which propagates through the equilibrium 
feedback laws and induces predictive distributions over future inputs and 
state evolution.

At each sampling instant $t_k$, the available discrete-time measurements 
$(\bm{x}(t_k), \dot{\bm{x}}(t_k), \bm{u}_1(t_k), \bm{u}_2(t_k))$ are used to 
construct a data pair $(\bm{\Phi}_i^{(k)}, \bm{y}_i^{(k)})$ via the feature-based 
representation of the HJB and feedback relations.  
Accounting for approximation error and deviations from exact optimality, each 
observation is modeled as
\begin{align}
    \bm{y}_i^{(k)} &= \bm{\Phi}_i^{(k)} \bm{W}_i^- + \bm{\eta}_i^{(k)} .
\end{align}

We adopt a Gaussian prior for $\bm{W}_i^-$,
\begin{align}
    \bm{W}_i^- \sim \pazocal{N}(\bm{m}_{0,i}, \bm{S}_{0,i}),
\end{align}
which compactly encodes prior structure (e.g., expected scale or sparsity) and 
quantifies initial uncertainty.  
The disturbance term $\bm{\eta}_i^{(k)}$ is modeled as
\begin{align}
    \bm{\eta}_i^{(k)} \sim \pazocal{N}(\mathbf{0}, \bm{\Sigma}_i),
\end{align}
reflecting the aggregation of several independent error sources in the HJB 
residual—feature-approximation error, numerical differentiation of 
$\dot{\bm{x}}(t_k)$, and modeling mismatch—whose combined effect is 
well-approximated as Gaussian by standard central-limit arguments.

Under these linear–Gaussian assumptions, the posterior distribution of 
$\bm{W}_i^-$ remains Gaussian and admits closed-form recursive updates, 
providing an efficient inference scheme suitable for online operation.  
The formulation here focuses on static objective parameters over a single 
interaction episode; extensions to time-varying objectives can be developed 
by placing a stochastic evolution model on $\bm{W}_i^-$.

\subsection{Conjugate Bayesian Update}

Let $\pazocal{D}_k$ denote all data up to time $t_k$. Then
\begin{align}
    \bm{W}_i^- \mid \pazocal{D}_k
        &\sim \pazocal{N}\!\left(\bm{m}_i^{(k)},\bm{S}_i^{(k)}\right),
\end{align}
where $\bm{m}_i^{(k)}$ and $\bm{S}_i^{(k)}$ are the posterior mean and covariance.

Given a new sample $(\bm{\Phi}_i^{(k+1)},\bm{y}_i^{(k+1)})$, the updates are:

\smallskip
\noindent\textit{1) Predicted output:}
\begin{align}
    \hat{\bm{y}}_i^{(k+1)}
        &= \bm{\Phi}_i^{(k+1)} \bm{m}_i^{(k)}.
\end{align}

\noindent\textit{2) Predictive covariance:}
\begin{align}
    \bm{S}_{y,i}
        &= \bm{\Phi}_i^{(k+1)} \bm{S}_i^{(k)}
           \bm{\Phi}_i^{(k+1)\!\top}
           + \bm{\Sigma}_i.
\end{align}

\noindent\textit{3) Posterior gain:}
\begin{align}
    \bm{K}_i
        &= \bm{S}_i^{(k)} \bm{\Phi}_i^{(k+1)\!\top} \bm{S}_{y,i}^{-1}.
\end{align}

\noindent\textit{4) Posterior mean:}
\begin{align}
    \bm{m}_i^{(k+1)}
        &= \bm{m}_i^{(k)}
           + \bm{K}_i\!\left(
               \bm{y}_i^{(k+1)}
               - \hat{\bm{y}}_i^{(k+1)}
           \right).
\end{align}

\noindent\textit{5) Posterior covariance:}
\begin{align}
    \bm{S}_i^{(k+1)}
        = (\bm{I} &- \bm{K}_i \bm{\Phi}_i^{(k+1)})\, \bm{S}_i^{(k)}
           (\bm{I} - \bm{K}_i \bm{\Phi}_i^{(k+1)})^\top \nonumber \\
           &  + \bm{K}_i \bm{\Sigma}_i \bm{K}_i^\top.
\end{align}

The update requires only the latest sample and the previous posterior; no
history stack is maintained.  The covariance $\bm{S}_i^{(k)}$ evolves in real
time, reflecting how strongly the data support the inferred objective and
highlighting inconsistencies with newly observed behavior through posterior
shifts or increases in uncertainty.

\subsection{Predictive Modeling and Monte Carlo Forecasting}
\label{sec:forecast_envelope}
Given the posterior $\pazocal{N}(\bm{m}_i^{(k)},\bm{S}_i^{(k)})$ over the
reduced parameter vector $\bm{W}_i^-$, the feedback law
$\bm{\mu}_i^\star(\bm{x};\bm{W}_i^-)$ in~\eqref{eq:feedback_linear_param}
induces a \emph{stochastic policy model} for player~$i$ over a prediction
rollout:
\[
    \bm{u}_i(t)
    = \bm{\mu}_i^\star\!\big(\bm{x}(t);\bm{W}_i^-\big),
    \qquad
    \bm{W}_i^- \sim \pazocal{N}(\bm{m}_i^{(k)},\bm{S}_i^{(k)}).
\]
This stochastic policy is not executed by the agent; it reflects our belief,
derived from data, about how the agent is likely to act under parameter
uncertainty.  
Consequently, the induced closed-loop evolution is stochastic only through
uncertainty in~$\bm{W}_i^-$.  
For any horizon $T>0$, the resulting distribution over trajectories is given
by the pushforward measure
\[
    \pazocal{P}_{\bm{x}(\cdot)}
        = \Gamma_\#\!
          \left( 
              \pazocal{N}(\bm{m}_i^{(k)},\bm{S}_i^{(k)})
          \right),
\]
where $\Gamma$ maps a parameter realization to its closed-loop trajectory on
$[0,T]$.

Since this distribution is analytically intractable, we approximate it via
Monte Carlo simulation.  
For each rollout $s=1,\dots,N_{\mathrm{mc}}$, we sample
$\bm{W}_i^{-(s)}$ from the posterior, generate the corresponding stochastic
policy model, and simulate the closed-loop dynamics to obtain trajectories
$\bm{x}^{(s)}(t)$ and $\bm{u}_i^{(s)}(t)$.
We use the Monte Carlo rollouts to construct a \emph{forecast envelope} for
the predicted closed-loop evolution.  
For $t=0:T$ and state dimension $n_x$, we introduce componentwise bounds
\begin{align}
\label{envelope}
    \pazocal{X}_t(\theta)
    = \{\bm{x}\in\mathbb{R}^{n_x} : \underline{x}_{t,j} \;\le\; x_j \;\le\; \overline{x}_{t,j},\;
       j=1{:}n_x\},
\end{align}
parameterized by
$\theta = \{\underline{x}_{t,j},\, \overline{x}_{t,j}\}_{t=0{:}T,\, j=1{:}n_x}
\in\mathbb{R}^d$, where $d=2(T+1)n_x$.
We focus on state envelopes for clarity, although the same construction
applies to the predicted input sequence $\bm{u}_i(t)$ obtained
from the Monte Carlo rollouts.

\begin{theorem}[Reliability of the Forecast Envelope]
\label{thm:forecast_tube}
Let $\theta^\star$ denote an envelope parameter obtained from the Monte
Carlo trajectories $\bm{x}^{(s)}(t)=\Gamma(\bm{W}_i^{-(s)})(t)$,
$s=1:N_{\mathrm{mc}}$, $t=0:T$.  
For any $\varepsilon,\beta\in(0,1)$, if the number of samples
$N_{\mathrm{mc}}$ satisfies
\begin{equation}
\label{eq:scenario_bound}
\sum_{i=0}^{d-1}
    \binom{N_{\mathrm{mc}}}{i}
    \varepsilon^{i}(1-\varepsilon)^{N_{\mathrm{mc}}-i}
\le \beta,
\end{equation}
then
\[
\mathbb{P}\!\left(V(\theta^\star)\le\varepsilon\right)\ge 1-\beta,
\]
where 
$V(\theta)
=\mathbb{P}_{\bm{W}_i^-}\{\exists\,t:\Gamma(\bm{W}_i^-)(t)
  \notin\pazocal{X}_t(\theta)\}$,
and the probability is with respect to the draw of the
$N_{\mathrm{mc}}$ posterior samples.
\end{theorem}

\begin{proof}
Each rollout is generated by independently sampling 
$\bm{W}_i^{-(s)}\sim\pazocal{N}(\bm{m}_i^{(k)},\bm{S}_i^{(k)})$
and propagating it through the closed-loop map $\Gamma$.  
Thus the constraints
$\bm{x}^{(s)}(t)\in\pazocal{X}_t(\theta)$ correspond to i.i.d.\ scenario
constraints associated with the uncertainty $\bm{W}_i^-$.  

A forecast envelope $\theta^\star$ can be obtained by solving the convex
scenario program
\[
\begin{aligned}
\theta^\star \in \arg\min_{\theta} \quad &
    \left\{ \sum_{t=0}^{T}\sum_{j=1}^{n_x} \big(\overline{x}_{t,j} - \underline{x}_{t,j}\big) \right\} \\
\text{s.t.}\quad &
    \bm{x}^{(s)}(t)\in\pazocal{X}_t(\theta),
    \qquad \forall\, s=1{:}N_{\mathrm{mc}},\; t=0:T .
\end{aligned}
\]
The decision variable enters these constraints only through affine
inequalities, and the objective is linear; thus the program is convex.
By the scenario approach for convex programs~\citep{CampiGaratti_convex},
any such solution $\theta^\star$ satisfies
\[
\mathbb{P}(V(\theta^\star)>\varepsilon)
\le 
\sum_{i=0}^{d-1}
    \binom{N_{\mathrm{mc}}}{i}
    \varepsilon^{i}(1-\varepsilon)^{N_{\mathrm{mc}}-i},
\]
and condition~\eqref{eq:scenario_bound} ensures this probability is at most
$\beta$.
\end{proof}

Theorem~\ref{thm:forecast_tube} establishes that the forecast envelope
$\{\pazocal{X}_t^\star\}_{t=0}^{T}$, constructed from the Monte Carlo
ensemble, enjoys a reliability guarantee with respect to
the posterior over $\bm{W}_i^-$.  
With confidence at least $1-\beta$ (over the randomness of the training
rollouts), the probability that a new trajectory sampled from the posterior
exits the envelope over the prediction horizon is at most $\varepsilon$.  
The binomial condition~\eqref{eq:scenario_bound} characterizes the
trade-off between ensemble size, violation level, and confidence: given
$(\varepsilon,\beta)$, the required number of Monte Carlo rollouts
$N_{\mathrm{mc}}$ is the smallest integer satisfying
\eqref{eq:scenario_bound}; conversely, for a fixed ensemble size
$N_{\mathrm{mc}}$ and a confidence level $1-\beta$, the certified violation
probability $\bar{\varepsilon}$ is the smallest value of $\varepsilon$ for
which~\eqref{eq:scenario_bound} holds.  
These guarantees are non-asymptotic and depend only on the envelope
parameterization dimension $d$ and the number of posterior samples, in
accordance with standard results from the scenario approach~\citep{CampiGaratti_convex}.  
The Monte Carlo ensemble and the resulting forecast envelope provide
empirical moments, nonparametric credible regions, and a
scenario-certified bound on predicted states and inputs, enabling
uncertainty-aware downstream tasks such as probabilistic constraint
checking, risk assessment, or anticipating opponent behavior.

\section{Simulations and Results}
\subsection{Linear--Quadratic (LQ) Inverse Game}
We first evaluate the proposed Bayesian inverse--game method on a 
linear--quadratic (LQ) differential game adapted from~\citep{ILQG_SharedControl}.\footnote{Code for the algorithm and simulations is available at \texttt{https://github.com/TUM-ITR/bayesian-agent-behavior}.}
In this setting, closed-form Nash equilibria can be computed efficiently 
(e.g., via the Lyapunov--iteration algorithm of~\citep{li1995lyapunov}), enabling 
direct comparison between the inferred parameters and ground truth.  
The system error dynamics are
\begin{align}
    \dot{\bm{x}}(t) = A\bm{x}(t) + B_1\bm{u}_1(t) + B_2\bm{u}_2(t),
\end{align}
with matrices
\begin{align*}
A = \begin{pmatrix}
    0 & 1 & -1 & 0 \\
    1 & 0 & 2  & 1 \\
    0 & -2 & 0 & 1 \\
    0 & 1 & 0 & -1
\end{pmatrix}, \;
B_1 = \tfrac{1}{2}\begin{pmatrix}
    0 & 1 \\ 0 & 0 \\ 0 & 0 \\ 1 & 0
\end{pmatrix}, \;
B_2 = \tfrac{1}{2}\begin{pmatrix}
    0 & 0 \\ 0 & 1 \\ 1 & 0 \\ 0 & 0
\end{pmatrix}.
\end{align*}

The nominal (mean) cost matrices are
\begin{align*}
    Q_{1,\mathrm{mean}} &= \mathrm{diag}(1,\,2/5,\,3,\,1), &
    R_{1,\mathrm{mean}} &= \mathrm{diag}(1,\,1), \\
    Q_{2,\mathrm{mean}} &= \mathrm{diag}(1,\,2/3,\,1,\,2), &
    R_{2,\mathrm{mean}} &= \mathrm{diag}(1,\,0.5).
\end{align*}

As in standard inverse-LQ formulations, the intrinsic scale ambiguity is removed
by treating $R_{i,[11]}$ as known. The remaining diagonal entries of $Q_i$ and
$R_i$ are sampled independently from Gaussians centered at their nominal values,
with variances
\begin{align*}
    \mathrm{Var}(Q_1) &= \mathrm{diag}(1,\,0.16,\,9,\,1), &
    \mathrm{Var}(R_{1,[22]}) &= 1, \\
    \mathrm{Var}(Q_2) &= \mathrm{diag}(1,\,4/9,\,1,\,4), &
    \mathrm{Var}(R_{2,[22]}) &= \tfrac{1}{4}.
\end{align*}
Thus, for all non-fixed entries,
\begin{subequations}
\label{distribs_1stex}
\begin{align}
    Q_{i,[jj]} &\sim \pazocal{N}\!\left(Q_{i,[jj],\mathrm{mean}},\; \mathrm{Var}(Q_i)_{jj}\right), \\
    R_{i,[22]} &\sim \pazocal{N}\!\left(R_{i,[22],\mathrm{mean}},\; \mathrm{Var}(R_i)_{22}\right).
\end{align}
\end{subequations}

The prior mean $\bm{m}_{i,0}$ is constructed from the nominal cost matrices
and the associated value-function solutions $V_i(\bm{x}) = \bm{x}^\top P_i \bm{x}$.
The corresponding quadratic feature map is
\[
    \phi_{V_i}(\bm{x})
    = [\,x_1^2,\; x_1 x_2,\; \dots,\; x_3 x_4,\; x_4^2\,]^\top .
\]

In the LQ setting, the known quadratic structure of the value function allows for
a reliable estimation of parameter variability. A Monte Carlo procedure using
$200$ samples drawn from~\eqref{distribs_1stex} is used to compute the associated
value functions and obtain a sample-based characterization of feature-parameter
variability. Together with the priors in~\eqref{distribs_1stex}, this yields the
prior $(\bm{m}_{i,0}, \bm{S}_{i,0})$.

We then run the online Bayesian inverse--game algorithm on data generated from
an LQ game initialized at
\[
    \bm{x}_0 = [\,3,\,-4,\,2,\,1.5\,]^\top.
\]
To capture representative closed-loop behavior in a range of task-relevant
configurations, the system is successively regulated toward three intermediate
targets:
\[
    \bm{x}_1 = \bm{0}^\top,\quad
    \bm{x}_2 = [\,1,-2,2,1\,]^\top,\quad
    \bm{x}_3 = [\,-2,1,3,-2\,]^\top,
\]
with data sampled at $\Delta t = 0.01\,\mathrm{s}$.

Figure~\ref{params_p1_LQ} shows the Bayesian prediction of the game parameters 
for Player~1. The posterior means converge rapidly to the true values, while the 
shaded $2\sigma$ region shrinks as more data are incorporated, indicating growing 
confidence in the inferred objective. Parameters that are more strongly excited 
by the trajectory tighten earlier, illustrating the data-adaptive nature of the 
online estimator.

Next, we demonstrate how the algorithm enables probabilistic predictions of 
adversarial behavior. We assume the role of Player~2 and estimate Player~1's 
objective using only the first $30$ data samples, representing an early-phase 
interaction where prior knowledge is limited. The resulting posterior is then 
used to predict Player~1's future inputs and the corresponding closed-loop 
state evolution.  
All prediction rollouts are initialized from the same state $\bm{x}_0$ as in the 
parameter-identification evaluation.
To visualize typical predicted behavior, we draw $N_{\mathrm{mc}} = 10000$ 
samples from the posterior and propagate them through the dynamics.  
Figure~\ref{input_preds} shows the predicted inputs for both control channels: 
the posterior mean closely follows the true input signals, while the $95\%$ 
credible regions—computed as pointwise empirical quantiles—expand in portions 
of the trajectory where Player~1's behavior is harder to predict.  
An analogous plot for the state evolution is provided in 
Figure~\ref{state_preds}.  
In both figures, we additionally overlay the scenario-certified forecast 
envelopes (dashed lines), allowing a direct comparison between empirical 
credible regions and guaranteed bounds.

It is worth noting that visually stable credible regions typically require only 
a modest number of Monte Carlo samples, whereas scenario-theoretic guarantees 
are more demanding.  
In this experiment, the prediction horizon spans $T = 6 \si{s}$ with a sampling 
time of $\Delta t = 0.1 \si{s}$, resulting in a long sequence of envelope 
parameters; obtaining nontrivial violation bounds therefore necessitates a 
sufficiently large number of rollout samples.  
Using the same $N_{\mathrm{mc}}$ Monte Carlo rollouts, we construct the forecast 
envelope introduced in Section~\ref{sec:forecast_envelope}.  
Applying Theorem~\ref{thm:forecast_tube} with confidence level $1-\beta = 0.99$ 
yields a certified violation probability of $\varepsilon = 0.0531$.  
That is, with $99\%$ confidence (over the randomness of the posterior 
sampling), a trajectory drawn from the posterior exits the forecast envelope 
with probability at most $5.31\%$.

In summary, the $95\%$ credible regions describe the typical spread of posterior 
predictions at each time step, while the scenario-certified forecast envelope 
provides a rigorous, distribution-free guarantee on the entire predicted 
trajectory.  
Both are derived from the same Monte Carlo ensemble and together offer a 
comprehensive characterization of uncertainty in Player~1's predicted behavior.

\begin{figure}[t]
    \includegraphics[width=0.48\textwidth]{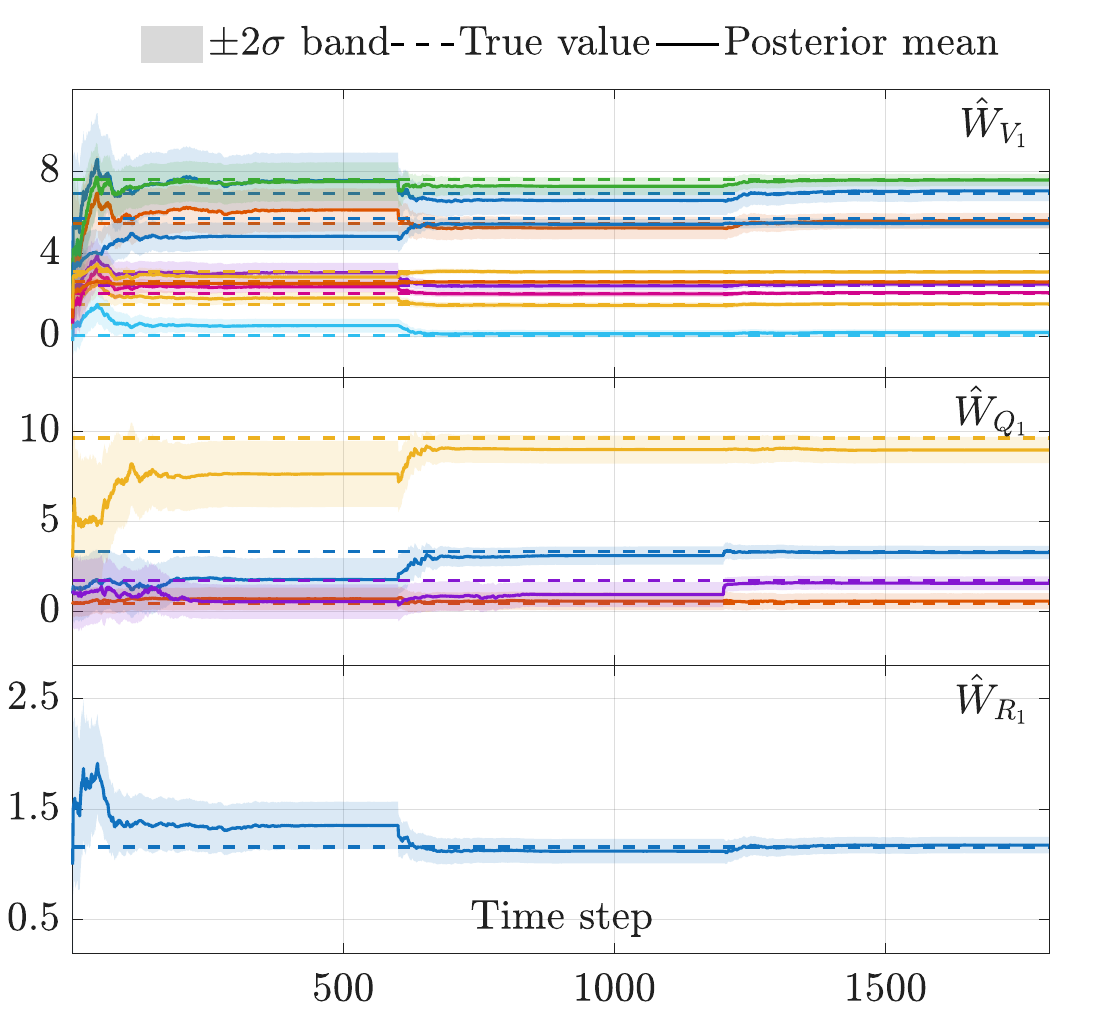}
    \vspace{-0.25cm}
    \caption{Convergence of the inferred game parameters for Player~1. 
    The dashed lines denote the true values, solid lines show the posterior means, 
    and the shaded regions represent the $\pm 2\sigma$ uncertainty. 
    The three panels correspond to value-function weights, state-cost weights, 
    and the unknown input-cost weight.}
    \label{params_p1_LQ}
\end{figure}

\begin{figure}[t]
    \includegraphics[width=0.48\textwidth]{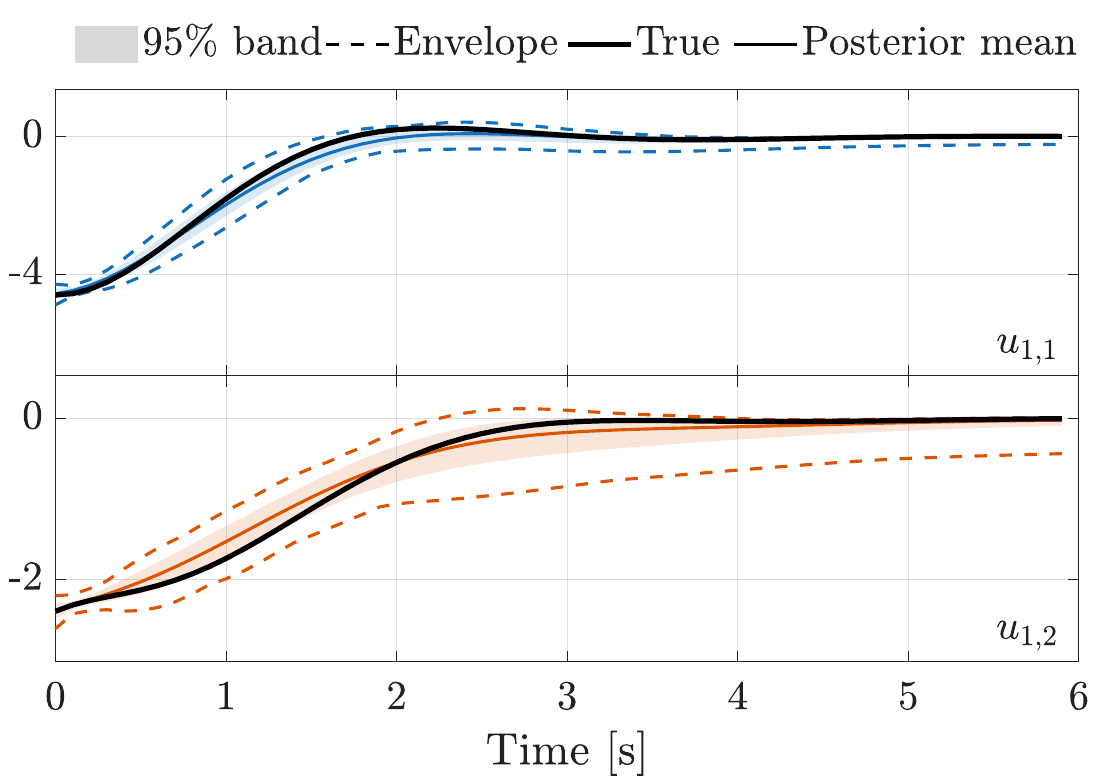}
    \vspace{-0.25cm}
    \caption{Predicted inputs of Player~1 under uncertain cost estimation. 
    True inputs (thick solid), posterior means (solid), $95\%$ credible regions (shaded), 
and scenario-certified forecast envelopes (dashed, colored) are shown for both control 
channels (blue and orange).}
    \label{input_preds}
\end{figure}

\begin{figure}[t]
    \includegraphics[width=0.48\textwidth]{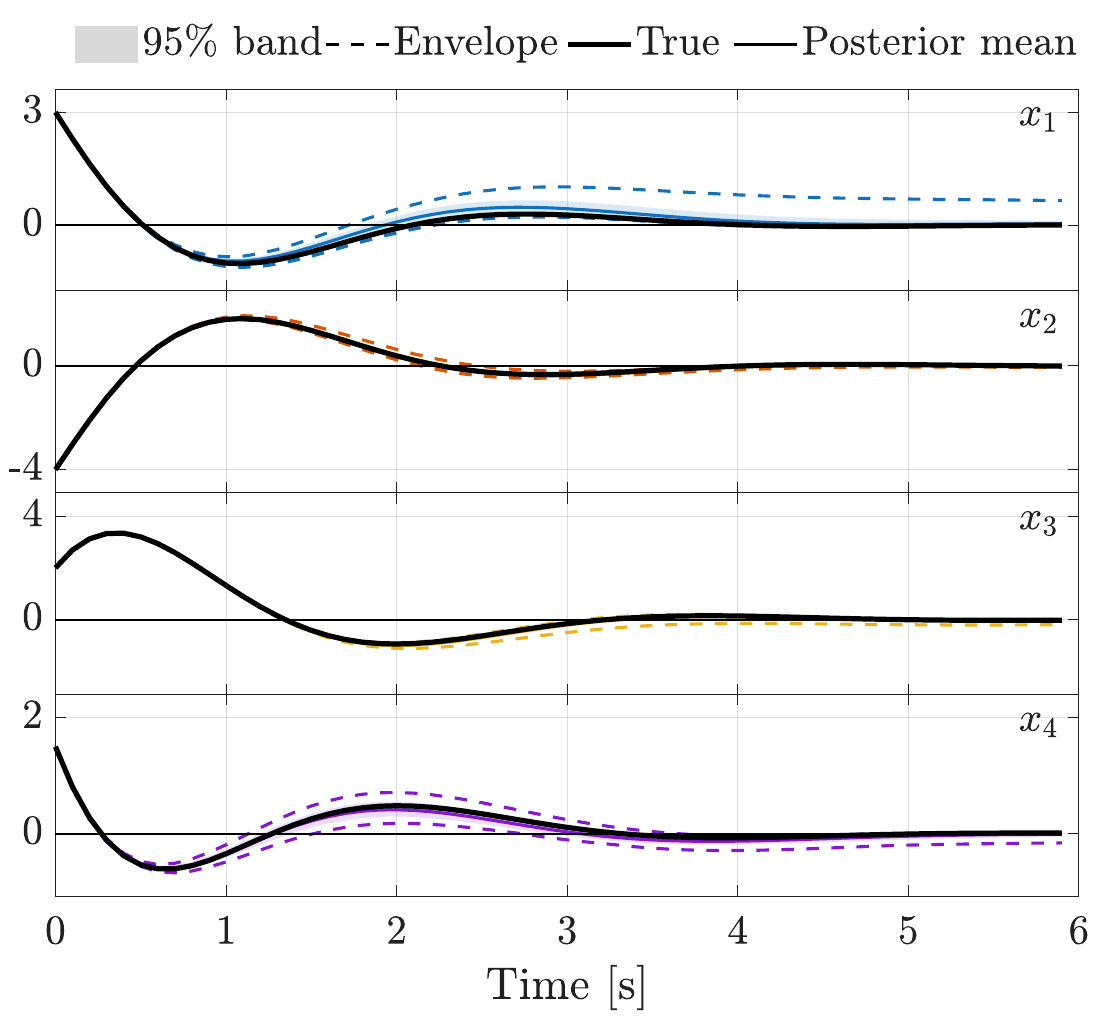}
    \vspace{-0.25cm}
    \caption{Predicted state trajectories under uncertainty in Player~1's behavior. 
True trajectories (thick solid), posterior means (solid), $95\%$ credible regions (shaded), 
and scenario-certified forecast envelopes (dashed, colored) are shown for all state components.}
    \label{state_preds}
\end{figure}

\subsection{Nonlinear Shared-Control Example}
We next consider a nonlinear one-dimensional differential game that reflects
characteristic features of shared-control scenarios in HRI. Unlike the LQ
example, the interaction dynamics are state-dependent and non-affine, yielding
configuration-dependent and highly nonuniform influence of the two players.  
The drift and input channels are
\begin{subequations}
\label{eq:nonlinear_dynamics_example}
\begin{align}  
    f(x) &= \frac{-0.5x - 0.3x^3}{b(x)}, \\[2pt]
    g_1(x) &= \frac{1.2 + 0.8\sin(1.5x)}{b(x)}, \\[2pt]
    g_2(x) &= \frac{1.0 - 0.7\sin(1.5x + \pi/3)}{b(x)},
\end{align}
\end{subequations}
where
\[
    b(x) = b_0 \bigl(1 + \beta x^2\bigr),
    \qquad b_0 = 2.0,\; \beta = 0.3,
\]
models a state-dependent “inertia’’ term.

The oscillatory structure of $(g_1,g_2)$ creates alternating regions in which
either the human or the robot dominates the effective control authority. Such
patterns are representative of physical shared-control settings where task
geometry or interaction constraints induce configuration-dependent
responsiveness. The resulting game is more challenging than the LQ case because
the players’ influence varies across the state space, and the HJB relations are
nonlinear even after fixing the running costs.

Each player uses a quadratic running cost
\begin{equation}
    \label{eq:NL_cost}
    \ell_i(x,u_i) = Q_i x^2 + R_i u_i^2, 
    \qquad R_1 = R_2 = 1,
\end{equation}
with unknown positive parameters $Q_i$. To generate realistic ground truth while
capturing prior variability, each $Q_i$ is sampled from a Gaussian distribution
with means
\begin{align}
    \label{eq:NL_state_costs}
    Q_{1,\mathrm{mean}} = 1.0, \qquad Q_{2,\mathrm{mean}} = 0.3.
\end{align}
and variances
\begin{align}
    \mathrm{Var}(Q_1) = 0.5 \qquad
    \mathrm{Var}(Q_2) = 0.18. 
\end{align}
For these sampled costs, the coupled nonlinear HJB equations associated with
\eqref{eq:nonlinear_dynamics_example} and \eqref{eq:NL_cost} are solved numerically on a
compact domain using a Legendre basis of order~10. The resulting value functions
$V_i^\star(x)$ are used both to generate expert trajectories and to benchmark
the Bayesian inverse-game identification.

We simulate two regulation tasks, starting from $x_0=4$ and $x_1=-4$, over an
$8\,\mathrm{s}$ horizon, collecting data at each time step. The Bayesian
inverse-game algorithm is then run online along the resulting trajectory.
Figure~\ref{NL_value-funcs} shows the estimated value functions after $500$
updates ($t=5\,\mathrm{s}$). The posterior means closely match the true value
functions on the portions of the state space visited by the system, while the
uncertainty remains higher in regions that receive little or no data (e.g.,
$x\le -4$ or $x\ge 4$). This illustrates the estimator’s ability to recover
nonlinear objective structure where information is available and to reflect
appropriate uncertainty elsewhere.
\begin{figure}[t]
    \includegraphics[width=0.48\textwidth]{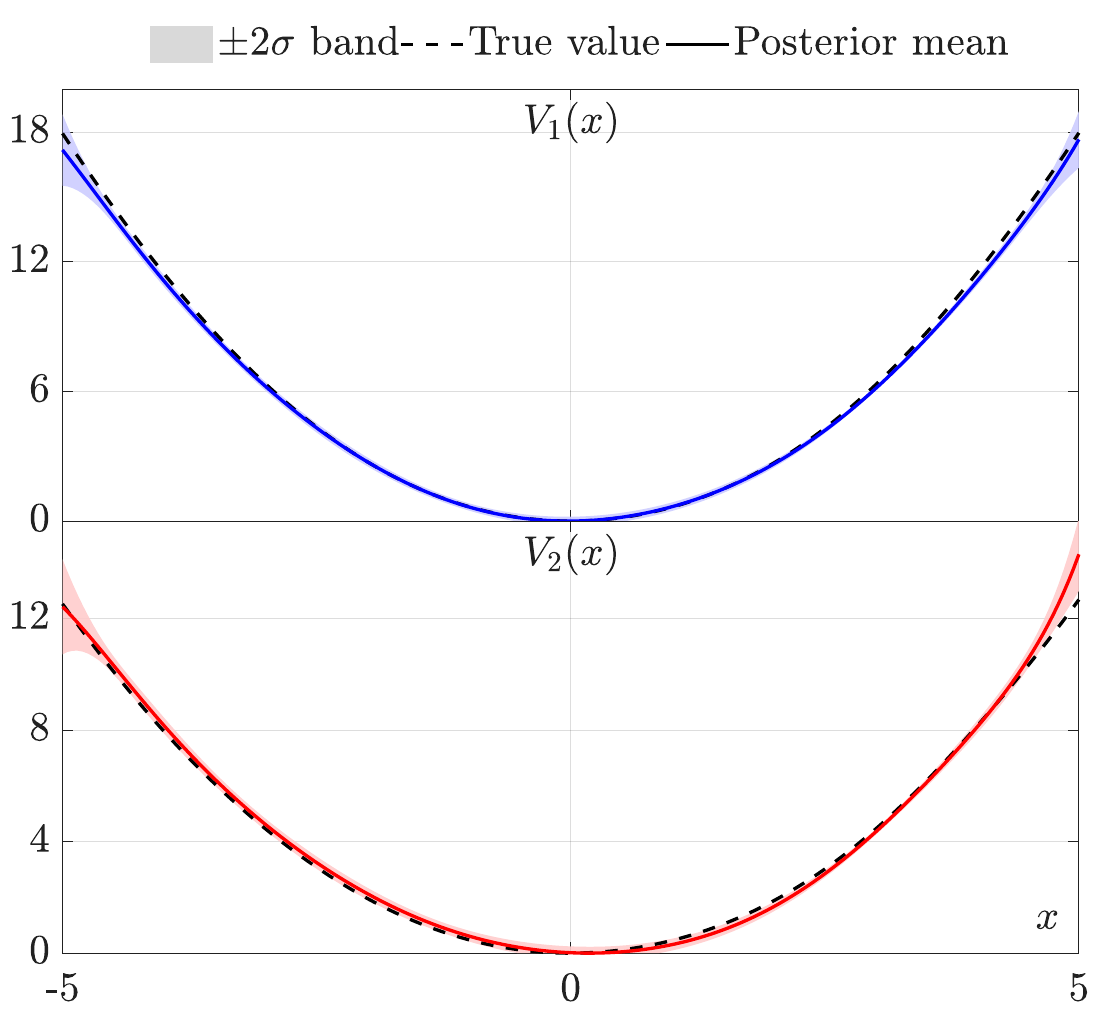}
    \vspace{-0.25cm}
    \caption{Estimated value functions for game (\ref{eq:nonlinear_dynamics_example}) after 500 updates ($t=5\,\mathrm{s}$).
    True value functions (dashed), posterior means (solid), and corresponding
    confidence regions (shaded). Blue: Player~1; Red: Player~2.}
    \label{NL_value-funcs}
\end{figure}

\section{Conclusion}
We introduced an online Bayesian framework for inverse differential games based
on Hamilton\allowbreak--\allowbreak Jacobi\allowbreak--\allowbreak Bellman relations.  
By combining feature-based HJB residuals with Gaussian priors, the method enables
uncertainty-aware inference, informative early predictions, and recursive updates
without history stacks.  
The framework applies to linear and nonlinear dynamics and supports general
value-function representations, making it suitable for adaptive, risk-aware
multi-agent interaction. Future work will explore time-varying objectives, richer nonparametric models,
and validation in human–robot interaction scenarios.


\section*{DECLARATION OF GENERATIVE AI AND AI-ASSISTED TECHNOLOGIES IN THE WRITING PROCESS}
During the preparation of this work, the authors used OpenAI's ChatGPT to assist with text revision and improvement of clarity. After using this tool, the authors reviewed and edited the content as needed and take full responsibility for the final version of the manuscript.

\bibliography{ifac_refs}             

\appendix       
\end{document}